%% file: main.tex
\documentclass[11pt,letterpaper]{article}

\usepackage{fullpage}

\usepackage{graphicx,subfigure}
\def\showauthornotes{2}

\def\showkeys{0}
\def\showdraftbox{1}
\def\showcolorlinks{1}
\def\usemicrotype{1}
\def\showfixme{1}

\def\arxivmode{0}
\def\fastmode{0}

\newcommand{\llangle}{\left\langle}
\newcommand{\rrangle}{\right\rangle}

\input{preamble/macros}

\input{preamble/macros-notes}
\usepackage[shortlabels]{enumitem}
\usepackage[normalem]{ulem}
\usepackage{stmaryrd}

\newcommand{\cmnt}[1]{}

\newcommand{\proj}{\mathrm{proj}}

\newcommand{\1}{\mathbb{1}}

\renewcommand{\mathbb}{\vvmathbb}

\newcommand{\spc}{\,:\,}
\newcommand{\vvD}{\vvmathbb{D}}
\newcommand{\vvW}{\vvmathbb{W}}
\newcommand{\cost}{\mathsf{cost}}
\newcommand{\scost}{\mathsf{serv}}
\newcommand{\mcost}{\mathsf{move}}
\newcommand{\flow}{\Lambda}

\newcommand{\rt}{\vvmathbb{r}}

\renewcommand{\dag}{\cD}

\title{Multiscale entropic regularization for MTS \\ on general metric spaces}

\author{Farzam Ebrahimnejad\thanks{\texttt{febrahim@cs.washington.edu}}\hspace{0.8in}  James R. Lee\thanks{\texttt{jrl@cs.washington.edu}} \vspace{0.1in}\\
{\small Paul G. Allen School of Computer Science \& Engineering} \\ {\small University of Washington}}
\date{}

\begin{document}

\maketitle

\vspace*{-0.4in}

\begin{abstract}
   We present an
   $O((\log n)^2)$-competitive algorithm for metrical task systems (MTS) on any $n$-point metric space
   that is also $1$-competitive for service costs.  This matches the competitive ratio achieved
   by Bubeck, Cohen, Lee, and Lee (2019) and the refined competitive ratios obtained by
   Coester and Lee (2019).  Those algorithms work by first randomly embedding the metric space
   into an ultrametric and then solving MTS there.  In contrast, our algorithm is cast as
   regularized gradient descent where the regularizer is a multiscale metric entropy defined directly on the metric space.
   This answers an open question of Bubeck (Highlights of Algorithms, 2019).
\end{abstract}

\begingroup
\hypersetup{linktocpage=false}
\setcounter{tocdepth}{2}
{\small
\tableofcontents}
\endgroup

\newpage

\section{Introduction}

Let $(X,d)$ be a finite metric space with $|X|=n > 1$. The Metrical Task
Systems (MTS) problem, introduced in \cite{BLS92} is defined as follows.
The input is a sequence $\langle c_t : X \to \R_+ \mid t = 1,2,\ldots\rangle$ of
nonnegative cost functions on the state space $X$.
At every time $t$, an online algorithm maintains
a state $\rho_t \in X$.

The corresponding cost is the sum of a {\em service cost} $c_t(\rho_t)$ and a
{\em movement cost} $d(\rho_{t-1}, \rho_t)$.
Formally, an {\em online algorithm} is a sequence of mappings
$\bm{\rho} = \langle \rho_1, \rho_2, \ldots, \rangle$
where, for every $t \geq 1$,
$\rho_t : (\R_+^X)^t \to X$ maps a sequence of cost functions $\langle c_1, \ldots, c_t\rangle$
to a state. The initial state $\rho_0 \in X$ is fixed. The {\em total cost of the algorithm $\bm{\rho}$ in servicing $\bm{c} = \langle c_t : t \geq 1\rangle$} is defined as the sum of the service and movement costs:
\begin{align*}
   \scost_{\bm{\rho}}(\bm{c}) &\seteq \sum_{t \geq 1} c_t\!\left(\rho_t(c_1,\ldots, c_t)\right)  \\
   \mcost_{\bm{\rho}}(\bm{c}) &\seteq \sum_{t \geq 1} d\!\left(\rho_{t-1}(c_1,\ldots, c_{t-1}), \rho_t(c_1,\ldots, c_t)\right) \\
   \cost_{\bm{\rho}}(\bm{c}) &\seteq \scost_{\bm{\rho}}(\bm{c}) + \mcost_{\bm{\rho}}(\bm{c}).
\end{align*}
The cost of the {\em offline optimum}, denoted $\cost^*(\bm{c})$, is the infimum of
$\sum_{t \geq 1} [c_t(\rho_t)+d(\rho_{t-1},\rho_t)]$ over {\em any} sequence
$\llangle \rho_t : t \geq 1\rrangle$ of states.

A {\em randomized online algorithm} $\bm{\rho}$  is said to be {\em $\alpha$-competitive}
if for every $\rho_0 \in X$, there is a constant $\beta > 0$ such that for all
cost sequences $\bm{c}$:
\[
\E\left[\cost_{\bm{\rho}}(\bm{c})\right] \leq \alpha \cdot \cost^*(\bm{c}) + \beta\,.
\]
Such an algorithm is said to be {\em $\alpha$-competitive for service costs and $\alpha'$-competitive
for movement costs} if there is a constant $\beta > 0$ such that for all
cost sequences $\bm{c}$:
\begin{align*}
   \E\left[\scost_{\bm{\rho}}(\bm{c})\right] &\leq \alpha \cdot \cost^*(\bm{c}) + \beta \\
   \E\left[\mcost_{\bm{\rho}}(\bm{c})\right] &\leq \alpha' \cdot \cost^*(\bm{c}) + \beta.
\end{align*}

For the $n$-point uniform metric,
a simple coupon-collector argument shows that the competitive ratio is $\Omega(\log n)$, and this is tight \cite{BLS92}. A long-standing conjecture is that this $\Theta(\log n)$ competitive ratio holds for an arbitrary $n$-point metric space.
The lower bound has almost been established \cite{BBM06,BLMN05};
for any $n$-point metric space, the competitive ratio is $\Omega(\log n / \log \log n)$.
Following a long sequence of works (see, e.g., \cite{Sei99,BKRS00,BBBT97,Bar96,FM03,FRT04}),
an upper bound of $O((\log n)^2)$ was shown in \cite{BCLL19}.

\paragraph{Competitive analysis via gradient descent}

Let us consider an equivalent fractional perspective on MTS where
the online algorithm maintains, at every point in time, a probability distribution $\mu_t \in \R_+^X$,
and we interpret the costs similarly as a vector $c_t \in \R_+^X$.
The cost of the algorithm is then given by
\[
   \sum_{t \geq 1} \left( \langle \mu_t, c_t\rangle + \vvW^1_X(\mu_{t-1},\mu_t)\right),
\]
where $\vvW^1_X$ is the $L^1$ transportation cost between two probability distributions on $(X,d)$.
This perspective is convenient, as now the state of the algorithm is given by
a point in the probability simplex $\Delta_X \subseteq \R_+^X$.

This yields a natural first algorithm for solving MTS:
\begin{equation}\label{eq:euclidean-proj}
   \mu_{t+1} \seteq \proj_{\Delta_X} \left(\mu_t -  \eta c_t\right),
\end{equation}
where $\eta > 0$ is some parameter we can choose and $\proj_{\Delta_X}$ denotes the Euclidean projection onto
the convex body $\Delta_X$.  Moreover, it gives a natural way of relating the 
cost incurred by the algorithm to the cost incurred by {\em any other} state $\nu \in \Delta_X$:
It is a basic exercise in convex geometry to show that
\begin{equation}\label{eq:euclidean-progress}
   \|\mu_{t+1}-\nu\|^2  - \|\mu_t - \nu\|^2 \leq \eta \langle c_t,\nu-\mu_t\rangle.
\end{equation}
In other words, if $\langle c_t,\mu_t\rangle > \langle c_t,\nu\rangle$, then $\mu_t$ approaches $\nu$
proportionally in the squared Euclidean distance.

Thus we cannot consistently incur more service cost than any fixed state.
This does not provide a competitive algorithm because there is, in general,
no convenient relationship between the Euclidean distance $\|\mu_{t}-\mu_{t+1}\|$ and the transportation distance $\vvW^1_X(\mu_t,\mu_{t+1})$.

But one can replace the Euclidean distance by any Bregman divergence $\vvD_{\Phi}$ associated to a strictly convex function $\Phi$.
Equivalently, we perform the projection \eqref{eq:euclidean-proj} in the local inner product
\[
   \langle u,v \rangle_{\mu_t} \seteq \langle \nabla^2 \Phi(\mu_t) u, v\rangle.
\]
Thus by choosing an appropriate geometry on $\Delta_X$, one can hope to obtain a competitive algorithm.
Such algorithms often go by the name {\em mirror descent} and the regularizer $\Phi$ is called the {\em mirror map} 
(we will often use the term {\em regularizer} interchangeably).

This framework is proposed in \cite{ABBS10,BCN14} and 
applied to the $k$-server problem in \cite{BCLLM18}, and to MTS in \cite{BCLL19} and \cite{CL19}.
In all these papers, the algorithms apply only to ultrametrics (equivalently, to hierarchically separated tree metrics (HSTs)).
In \cite{BCLL19}, mirror descent is used to analyze the algorithm on weighted stars, and these
algorithms are glued together in an ad-hoc way to handle HSTs.  In \cite{CL19}, stronger bounds
(known as ``refined guarantees'')
are obtained by finding an appropriate regularizer on arbitrary HSTs.
In both cases, general finite metric spaces are then handled via random embeddings into HSTs.

In the present work, we apply this method directly to MTS on general metric spaces and match the best-known competitive ratio.
Previously, it was unknown how to achieve any $\poly(\log n)$ competitive ratio for general metric spaces
using mirror descent and achieving this was posed as an open problem by Bubeck\footnote{Posed in his talk at HALG 2019.}.

We consider this an important step in advancing the underlying philosophy.  Note that past approaches to MTS
have involved a series of ad-hoc, complicated algorithms, along with clever potential function analyses.
In contrast, in the mirror descent approach, once one specifies a convex body and a regularizer,
both the algorithm and the method of analysis fall out naturally.  Indeed, the most subtle part of competitive analysis
lies in connecting the cost an online algorithm incurs to the cost of some offline optimum,
and this is done entirely through the general Bregman divergence analog of \eqref{eq:euclidean-progress}, which becomes
\[
   \vvD_{\Phi}(\nu \dmid \mu_{t+1}) - \vvD_{\Phi}(\nu \dmid \mu_t) \leq \langle c_t, \nu - \mu_t\rangle.
\]

\section{The multiscale noisy metric entropy}

To obtain $\poly(\log n)$-competitive algorithms for MTS, previous approaches \cite{BCLL19,CL19} employ a regularizer
that can be cast as a multiscale entropy for probability distributions on an underlying tree metric.
To handle general metric spaces, we will consider probability distributions on a lifted convex body
that is specified by a directed ayclic graph whose sinks are the points of $(X,d)$.
See \pref{fig:dag} for a pictoral representation when the metric space is a path.

\paragraph{The hierarchical flow DAG}

\begin{figure}[h]
      \begin{center}
         \includegraphics[width=8cm]{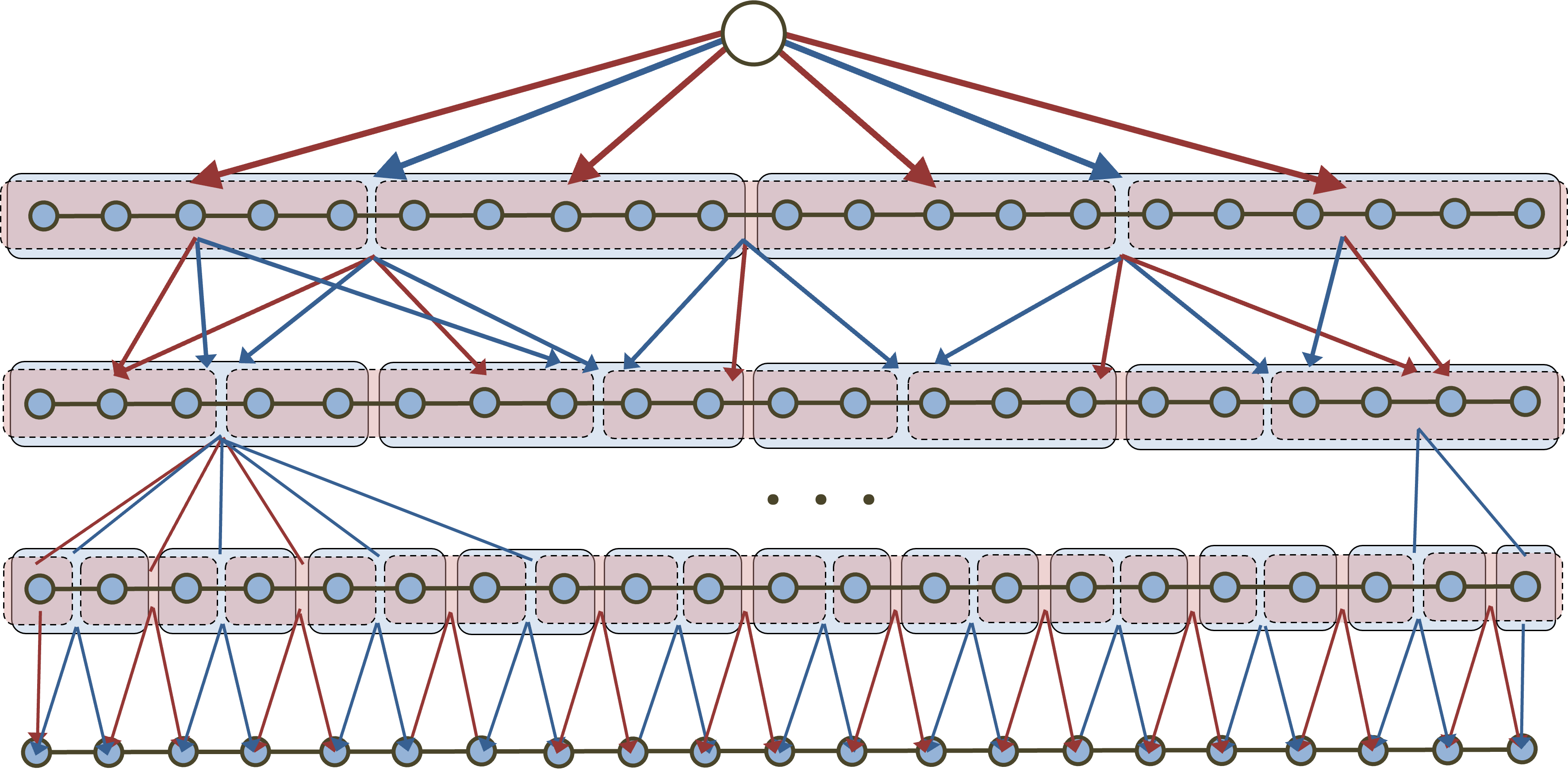}
         \caption{A hierarchical flow DAG over the path\label{fig:dag}}
   \end{center}
\end{figure}

Consider a finite set $X$ and a directed ayclic weighted graph $\dag = (V,A)$ with $X \subseteq V$
and such that
\begin{enumerate}[(i)]
   \item $\dag$ has a single source $\rt \in V$, and
   \item The set of sinks in $\dag$ is $X$.
\end{enumerate}
We say that {\em $\dag$ is a DAG over $X$.}
In what follows, we use the notation $\R_+ \seteq \{ x \in \R : x \geq 0 \}$ and $\R_{++} \seteq \{x \in \R : x > 0 \}$.
For an arc $(u,v) \in A$, we will often use the shorthand $uv$.

A vector $F \in \R_+^A$ is called a {\em flow in $\dag$} if holds that
\begin{equation}
\label{eq:flow-def}
	   \sum_{v \spc uv \in A} F_{uv} = \sum_{v \spc vu \in A} F_{vu},\quad \forall u \in V \setminus (X \cup \{\rt\}).
\end{equation}
For a flow $F$ and $u \in V \setminus X$, define $F_u \seteq \sum_{v \spc uv \in A} F_{uv}$.  For a sink $x \in X$, we define
$F_x \seteq \sum_{u \spc ux \in A} F_{ux}$ as the flow into $x$.
Say that $F$ is a {\em unit flow in $\dag$} if $F_{\rt} = 1$, and let $\cF_{\dag} \subseteq \R_+^{A}$ denote the convex
set of all unit flows in $\dag$.

A {\em (directed) path $\gamma$ in $\dag$} is a sequence $\gamma = \langle u_1 u_2, u_2 u_3, \ldots, u_{m-1} u_m\rangle$ with $u_i u_{i+1} \in A$ for each $i \in \{1,\ldots,m-1\}$.  We will occasionally also specify a path as a sequence of vertices.  We use $\bar{\gamma}$ to denote the final vertex $u_m$ of $\gamma$.
Let $\cP_{\dag}$ denote the set of all paths in $\dag$ from $\rt$ to some sink.

\paragraph{The multiscale entropy}

Let $\omega \in \R_{++}^A$ denote a vector of nonnegative arc lengths that are decreasing along paths,
i.e., such that $\omega_{uv} > \omega_{vw}$ whenever $uv,vw \in A$.
Let $\theta \in \R_{++}^A$
specify a probability distribution on the edges leaving every vertex, i.e.,
\begin{equation}
\label{eq:theta-distr}
	   \sum_{v \spc uv \in A} \theta_{uv} = 1,\qquad \forall u \in V \setminus X.
\end{equation}

Define the associated values
\begin{align}
   \eta_{uv} &\seteq 1 + \log(1/\theta_{uv}) \label{eq:eta-def} \\
   \delta_{uv} &\seteq \theta_{uv}/\eta_{uv}. \label{eq:delta-def}
\end{align}
We refer to the triple
$\hat{\dag} \seteq (\dag,\omega,\theta)$ as a {\em marked DAG}. For a given normalization parameter  $\kappa > 0$, such a marked DAG
yields a multiscale entropy functional 
$\Phi_{\hat{\dag}} : \cF_\dag \to \R_+$ defined by
\[
   \Phi_{\hat{\dag}}(F) \seteq \frac{1}{\kappa} \sum_{uv \in A} \frac{\omega_{uv}}{\eta_{uv}} \left(F_{uv} + \delta_{uv} F_u\right) \log \left(\frac{F_{uv}}{F_u} + \delta_{uv}\right).
\]
One can consult \cite{CL19} for a detailed discussion of multiscale entropies of this form on HSTs.

\paragraph{Two notions of depth}

We define two notions of depth associated to $\hat{\dag}$.  The first is the {\em combinatorial depth $\Delta_0(\dag)$}
which is the maximum number of arcs in any path from $\rt$ to some sink $X$.
For $\gamma \in \cP_{\dag}$, let us define
\begin{equation}
\label{eq:theta-prod-def}
	   \theta(\gamma) \seteq \prod_{uv \in \gamma} \theta_{uv},
\end{equation}
and let the {\em information depth} be defined as
\[
\Delta_I(\hat{\dag}) \seteq \max_{\gamma \in \cP_{\dag}} \log(1/\theta(\gamma)).
\]
Note that $\theta(\cdot)$ induces a probability distribution on $\cP_{\dag}$, and as clearly for $\gamma \in \cP_{\dag}$ it holds that $\theta(\gamma) \geq e^{-\Delta_I(\hat{\dag})}$ we have
\begin{equation}
\label{eq:Delta-I-PD-bound}
	\log |\cP_{\dag}| \leq \Delta_I(\hat{\dag}).
\end{equation}

\subsection{Mirror descent dynamics}
\label{sec:mirror-descent-dynamics}
Let us now fix a marked DAG $\hat{\dag}$ and take $\Phi \seteq \Phi_{\hat{\dag}}$.
We seek to define a continuous path $F : [0,\infty] \to \cF_\dag$ that represents the dynamics of
projected vector flow in response to a continuous path $c(t) \in \R_+^X$ of costs arriving at the points of $X$.

A natural Euclidean flow would be specified heuristically by
\[
   F(t+dt) = \proj_{\cF_\dag} \left(F(t) - c(t) \,dt\right),
\]
where for $v \in \R^A$, we define $\proj_{\cF_\dag}(v)$ as the unique point of $\cF_\dag$ with
minimal Euclidean distance to $v$.
In other words, we move a little in the direction $-c(t)$ and then project back to the feasible region $\cF_\dag$.

Instead, we will define our dynamics using the Bregman projection $\proj_{\cF_\dag}^{\Phi}$ associated to
our multiscale entropic regularizer, where
\[
   \proj_{\cF_{\dag}}^{\Phi}(v) \seteq \argmin \left\{ \vvD_{\Phi}\left(v' \dmid v\right) : v' \in \cF_\dag \right\},
\]
and
\[
   \vvD_{\Phi}\left(v' \dmid v\right) \seteq \Phi(v') - \Phi(v) - \langle \nabla \Phi(v), v'-v\rangle
\]
is the Bregman divergence associated to $\Phi$.

One can show that if $c(t)$ is continuous, then there is a path $F : [0,\infty) \to \cF_{\dag}$
for which the following dynamics are well-defined (for almost every $t \in [0,\infty)$):
\[
   F(t+dt) = \proj^{\Phi}_{\cF_\dag} \left(F(t) - c(t) \,dt\right)
\]

This path further satisfies (for almost all $t \in [0,\infty)$) the system of partial differential equations given by
\begin{equation}\label{eq:flow-dynamics}
   \partial_t \left(\frac{F_{uv}(t)}{F_u(t)}\right) = 
   \kappa \frac{\eta_{uv}}{\omega_{uv}} \left(\frac{F_{uv}(t)}{F_u(t)}+\delta_{uv}\right) \left(\beta_u(t) - \hat{c}_{uv}(t)\right), \quad uv \in A,
\end{equation}
where $\hat{c}_{uv}(t) = \1_{\{F_{uv}(t) > 0\}} c_{v}(t) $ if $v \in X$, and otherwise
\begin{equation}
\label{eq:hatc-def}
   \hat{c}_{uv}(t) = \1_{\{F_{uv}(t) > 0 \}} \sum_{w \spc vw \in A} \frac{F_{vw}(t)}{F_v(t)} \hat{c}_{vw}(t),
\end{equation}
and $\beta_u(t)$ is the unique value that guarantees
\[
   \partial_t \sum_{v \spc uv \in A} \frac{F_{uv}(t)}{F_u(t)} = 0,
\]
i.e.,
\[
   \beta_u(t)  =
   \frac{\sum_{v \spc uv \in A} \frac{\eta_{uv}}{\omega_{uv}} \left(\frac{F_{uv}(t)}{F_u(t)}+\delta_{uv}\right) \hat{c}_{uv}(t)}{\sum_{v \spc uv \in A} \frac{\eta_{uv}}{\omega_{uv}} \left(\frac{F_{uv}(t)}{F_u(t)}+\delta_{uv}\right)}.
\]

Here we express the algorithm in continuous time for conceptual simplicity; its evolution is completely specified
by the regularizer $\Phi_{\hat{\dag}}$ and the costs $c(t)$.
But the existence of a solution to \eqref{eq:flow-dynamics} is derived from the limit of discrete-time algorithms
in \pref{sec:discrete-time}.

\subsection{Metric compatibility}
\label{sec:metric-compat}

To analyze the algorithm specified by \eqref{eq:flow-dynamics} on a metric space $(X,d)$, we need additionally that $\hat{\dag}=(\dag,\omega,\theta)$
is compatible with the geometry of $(X,d)$.
Suppose that $\hat{\dag}$ is a marked DAG over $X$.
Say that {\em $\hat{\dag}$ is $\tau$-geometric} if it holds that
for every pair of consecutive arcs $uv,vw \in A$, we have $\omega_{uv} \geq \tau \omega_{vw}$.

Let us define a metric on $\cP_{\dag}$ as follows:  Suppose $\gamma_1,\gamma_2 \in \cP_{\dag}$ and let $u \in V$
be the first vertex at which they diverge, i.e., at which $uv_1 \in \gamma_1, uv_2 \in \gamma_2$ and $v_1 \neq v_2$.
Define the distance
\[
   \dist_{\hat{\dag}}(\gamma_1,\gamma_2) \seteq \max(\omega_{uv_1}, \omega_{uv_2}).
\]
One can check that this gives a metric on $\cP_{\dag}$ since the arc lengths are decreasing
along source-sink paths.  In fact, this defines an ultrametric on $\cP_{\dag}$.

Say that $\hat{\dag}$ is {\em $\e$-expanding (with respect to $(X,d)$)}
if for every pair $\gamma_1,\gamma_2 \in \cP_{\dag}$, it holds that
\[
   \dist_{\hat{\dag}}(\gamma_1,\gamma_2) \geq \e d(\bar{\gamma}_1,\bar{\gamma}_2),
\]
where we recall that $\bar{\gamma}_1,\bar{\gamma}_2 \in X$ are the endpoints of $\gamma_1$ and $\gamma_2$, respectively.

We may extend $\dist_{\hat{\dag}}$ to a distance on $\cF_{\dag}$ by defining $\vvW^1_{\hat{\dag}}(F,F')$ as the $L^1$-transportation
cost between $F,F' \in \cF_{\dag}$ with the underlying metric $\dist_{\hat{\dag}}$, noting that $F$ and $F'$ can be viewed
as probability distributions on $\cP_{\dag}$.

Say that {\em $\hat{\dag}$ is $L$-Lipschitz (with respect to $(X,d)$)} if for every path $x_1,x_2,\ldots,x_m \in X$,
there is a sequence of flows $F^{(1)},F^{(2)},\ldots,F^{(m)} \in \cF_{\dag}$ such that:
\begin{enumerate}
   \item $F^{(i)}$ is a unit flow to $x_i$ for every $i=1,2,\ldots,m$.
   \item It holds that
      \[
      	         \sum_{i=1}^{m-1} \vvW_{\hat{\dag}}^1(F^{(i)}, F^{(i+1)}) \leq L \sum_{i=1}^{m-1} d(x_i,x_{i+1}).
      	         \]
\end{enumerate}

Our main result follows from the next two theorems, which are proved in \pref{sec:discrete-time}
and \pref{sec:dag}, respectively.

\begin{theorem}
\label{thm:main-1}
   Suppose $(X,d)$ is a metric space and $\hat{\dag}$ is a $\tau$-geometric marked DAG over $X$,
   for some $\tau \geq 4$.
   If $\hat{\dag}$ is $\e$-expanding and $L$-Lipschitz with respect to $(X,d)$,
   then  for $\kappa = 6L$, the MTS algorithm specified by \eqref{eq:flow-dynamics}
   is $1$-competitive for service costs, and
   $O\!\left(\frac{L}{\e} \left(\Delta_0(\dag) + \Delta_I(\hat{\dag})\right)\right)$-competitive for movement costs.
\end{theorem}

\begin{theorem}
\label{thm:main-2}
   For every $n$-point metric space $(X,d)$, there is a $12$-geometric marked DAG $\hat{\dag}$ over $X$ that is $1$-expanding
   and $O(\log n)$-Lipschitz, and moreover satisfies
   \[
      \Delta_0(\dag)+\Delta_I(\hat{\dag}) \leq O(\log n).
   \]
\end{theorem}

\section{Construction of a compatible DAG over $(X,d)$}
\label{sec:dag}

In \pref{sec:dag-construction}, we present the main construction of
a marked DAG $\hat{\dag}$ whose vertices are net points at every scale.
Achieving the crucial property $\Delta_{I}(\hat{\dag}) \leq O(\log n)$ requires
choosing the net points and the arcs of $\dag$ carefully.
In \pref{sec:distortion}, we argue that $\hat{\dag}$ is $\e$-expanding and $L$-Lipschitz
for $\e = 1$ and $L \leq O(\log n)$.
It may not be that $\Delta_0(\dag) \leq O(\log n)$, but in \pref{sec:compress}
we give a generic way of obtaining this property while leaving the
other essential properties intact.

\subsection{Hierarchical nets}
\label{sec:dag-construction}

Fix an $n$-point metric space $(X,d)$ and assume, without loss of generality, that $\diam(X)=1$.
Define $\e \seteq \min \{ d(x,y) : x,y \in X \}$ and $K \seteq 1+\lceil \log_{\tau} (1/\e)\rceil$.

\paragraph{Construction of nets}

Consider a parameter $\eta > 0$.  We construct an $\eta$-net $N \subseteq X$ inductively as follows.
Define $N_0 \seteq \emptyset$ and for $j \geq 1$, inductively define the set
\[
   S_j \seteq X \setminus B_X(N_{j-1}, \eta).
\]
If $S_j = \emptyset$, then we take $N \seteq N_{j-1}$.  Otherwise, let $x_j \in S_j$ be
a point that maximizes $|B_X(x, \eta/3)|$ among $x \in S_j$ and define $N_j \seteq N_{j-1} \cup \{x_j\}$.

\begin{lemma}\label{lem:net}
   The set $N\subseteq X$ is an $\eta$-net with the property that for any set $W \subseteq X$, if
   \[
      x^* \in \argmax \left\{ |B_X(y, \eta/3)| : y \in N \cap B_X(W, 1.5 \eta) \right\},
   \]
   then
   \[
      |B_X(x^*, \eta/3)| \geq \max \left\{ |B_X(w,\eta/3)| : w \in W \right\}
   \]
\end{lemma}

\begin{proof}
   Suppose $x_j \in N$ is the element with $j$ minimal such that $B_X(x_j, 1.5 \eta) \cap W \neq \emptyset$.
   Then
   \[
      \left(B_X(x_1, \eta) \cup \cdots \cup B_X(Q_{j-1}, \eta)\right) \cap B_X(W, \eta/3) = \emptyset,
   \]
   and hence by the greedy selection procedure,
   \begin{align*}
      |B_X(x_j,\eta/3)| &= |B_X(x^*, \eta/3)| \\
      |B_X(x_j, \eta/3)| &\geq \max \left\{ |B_X(w,\eta/3)| : w \in W \right\},
   \end{align*}
   completing the proof.
\end{proof}

Denote $\tau \seteq 12$.
For each $k \in \{0,1,\ldots,K\}$, let $U_k$ denote a $\tau^{-k}$-net that satisfies \pref{lem:net} with $\eta = \tau^{-k}$.
We now construct a DAG $\dag = (V,A)$ with $V \seteq \left\{ (u,k) : u \in U_k, k \in \{0,1,\ldots,K\}\right\}$.
For $k \in \{0,1,\ldots,K-1\}$, let $A_k$ denote the collection of pairs
$(u,u')$ for every $u \in U_k$ and $u' \in U_{k+1}$ satisfying:
\begin{align}
   d(u,u') &\leq 4 \tau^{-k} \label{eq:close} \\
   |B_X(u, \tau^{-k}/3)| &\geq \max \left\{ |B_X(w,\tau^{-k}/3)| : w \in B_X(u', 6 \tau^{-(k+1)})\right\}\label{eq:strongly-viable}.
\end{align}

   We define $A \seteq \bigcup_{k=0}^{K-1} \left\{ \left((u,k),(u',k+1)\right) : (u,u') \in A_k \right\}$, and
   \[
      \omega_{(u,k) (u',k+1)} \seteq 10\tau^{-k}.
   \]
Since $U_K = X$, we can identify the sinks in $\dag$ with the points of $X$.
We take $\rt \seteq (u,0)$, where $U_0 = \{u\}$.

\begin{observation}
\label{obs:coverage}
Suppose that $(u,k) \in U_k$ and $(x,K) \in V$ is reachable in $\dag$ from $(u,k)$.
   Then 
   \[
      d(u,x) \leq 4 \tau^{-k} + 4 \tau^{-(k+1)} + \cdots + 4 \tau^{-(K-1)} < 5 \tau^{-k}.
   \]
\end{observation}

For a set $S \subseteq X$ and $k \in \{0,1,\ldots,K\}$, define
\begin{equation}\label{eq:selector}
   \f_k(S) \seteq \argmax \left\{ |B_X(y,\tau^{-k}/3)| : y \in B_X(S, 2\tau^{-k}) \cap U_k \right\}.
\end{equation}
We will require the following fact later.

\begin{lemma}\label{lem:selector}
   Consider a set $S \subseteq X$ with $\diam_X(S) \leq 2 \tau^{-k}$.
   If $u' \in S \cap U_{k+1}$, then $(\f_k(S),u') \in A_k$.
\end{lemma}

\begin{proof}
   Denote $u \seteq \f_k(S)$.
   Since $u' \in S$ and $u \in B_X(S, 2 \tau^{-k})$, it holds that $d(u,u') \leq 4 \tau^{-k}$,
   and therefore \eqref{eq:close} is satisfied.
   Now denote $W \seteq B_X(u', 6 \tau^{-(k+1)})$.
   Then $B_X(W, 1.5 \tau^{-k}) \subseteq B_X(S, 2 \tau^{-k})$, 
   hence \pref{lem:net} implies that
   \[
      |B_X(u, \tau^{-k}/3)| \geq \max \left\{ |B_X(w,\tau^{-k}/3)| : w \in W \right\},
   \]
   which shows that \eqref{eq:strongly-viable} is satisfied as well.
\end{proof}

For $k \in \{0,1,\ldots,K-1\}$ and $(u,u') \in A_k$, we define
\begin{equation}\label{eq:theta-def}
   \theta_{(u,k),(u',k+1)} \seteq \frac{|B_X(u', \tau^{-(k+1)}/3)|}{\sum_{w : (u,w) \in A_k} |B_X(w, \tau^{-(k+1)}/3)|}.
\end{equation}

\begin{claim}
		\label{claim:delta-bound}
	It holds that
		\[
         \sum_{w: (u,w) \in A_k} |B_X(w, \tau^{-(k+1)}/3)| \leq |B_X(u, 6 \tau^{-k})|\,.
	\]
	\end{claim}
	\begin{proof}
      Since the elements of $U_{k+1}$ form a $\tau^{-(k+1)}$-net, the balls $\left\{ B_X(w,\tau^{-(k+1)}/3) : (u,w) \in A_k \right\}$ are pairwise disjoint.
      Furthermore, by \pref{obs:coverage},
      every such ball is contained in 
      \[
         B_X(u, 5 \tau^{-k} + \tau^{-(k+1)}/3) \subseteq B_X(u, 6 \tau^{-k}).\qedhere
      \]
	\end{proof}
	
\begin{lemma}
\label{lem:root-leaf-sum-bound}
It holds that $\Delta_I(\dag) \leq 3 \log n$, i.e.,
for every path $\gamma \in \cP_{\dag}$,
\[
   \sum_{uv \in \gamma} \log (1/\theta_{u,v}) \leq 3 \log n.
\]
\end{lemma}
\begin{proof}
   Consider a path $\gamma = \langle (u_0, 0), (u_1,1),\ldots,(u_K, K) \rangle$.
   From the definition \eqref{eq:theta-def} and \pref{claim:delta-bound}, it holds that
\begin{align}
   \label{eq:sum-bound-1}
   \sum_{k=0}^{K-1} \log \left(1/\theta_{(u_k,k) (u_{k+1},k+1)}\right) \leq \sum_{k=0}^{K-1} \log \frac{|B_X(u_k, 6 \tau^{-k})|}{|B_X(u_{k+1}, \tau^{-(k+1)}/3)|}.
\end{align}

Let us denote $\ell \seteq u_K$.
By \pref{obs:coverage}, it holds that $d(\ell, u_k) \leq 6\tau^{-k}$ for $0 \leq k \leq K$.
Therefore,
\begin{equation}
\label{eq:sum-bound-2}
	B_X(u_k, 6 \tau^{-k}) \subseteq B_X(\ell, 12 \tau^{-k}).
\end{equation}
Furthermore since $(u_k,u_{k+1}) \in A_k$,
by \pref{eq:strongly-viable}, we have
\begin{equation}
\label{eq:sum-bound-3}
|B_X(u_k, \tau^{-k}/3)| \geq \max \left\{ |B_X(w, \tau^{-k}/3)| : w \in B_X(u_{k+1}, 6 \tau^{-(k+1)}) \right\} \geq |B_X(\ell, \tau^{-k}/3)|,
\end{equation}
since $d(\ell, u_{k+1}) \leq 6 \tau^{-(k+1)}$.

By combining \pref{eq:sum-bound-1}--\pref{eq:sum-bound-3}, we obtain
\[
   \sum_{k=0}^{K-1} \log \left(1/\theta_{(u_k,k) (u_{k+1},k+1)}\right) 
   \leq \sum_{k=0}^{K-1} \log \frac{|B_X(\ell, 12 \tau^{-k})|}{|B_X(\ell, \tau^{-(k+1)}/3)|}
   \leq \sum_{k=0}^{K-1} \log \frac{|B_X(\ell, \tau^{-(k-1)})|}{|B_X(\ell, \tau^{-(k+2)})|} \leq 3 \log n,
\]
where we used $\tau = 12$ in the penultimate inequality.
\end{proof}
The above result together with \pref{eq:Delta-I-PD-bound} yield the following.
\begin{corollary}
\label{cor:root-leaf-path-count}
It holds that $|\cP_{\dag}| \leq n^3$.
\end{corollary}

%
\subsection{Distortion analysis}
\label{sec:distortion}

\begin{lemma}
   \label{lem:expanding}
   It holds that $\hat{\dag}$ is $1$-expanding with respect to $(X,d)$. 
\end{lemma}

\begin{proof}
 Suppose that $\gamma_1,\gamma_2 \in \cP_{\dag}$ and let $u \in V$ be the first
   vertex for which $uv_1 \in \gamma_1$ and $uv_2 \in \gamma_2$ with $v_1 \neq v_2$.
   If $u = (x,k)$, 
   then $\omega_{u v_1} = \omega_{u v_2} = 10\tau^{-k}$ and so
   $\dist_{\hat{\dag}}(\gamma_1,\gamma_2) = 10\tau^{-k}$.
   Moreover, by \pref{obs:coverage} we have
   \[
      d(\bar{\gamma}_1,\bar{\gamma}_2) \leq d(\bar{\gamma}_1,x) + d(\bar{\gamma}_2,x) \leq 10 \tau^{-k},
   \] 
   completing the proof.
\end{proof}

For a partition $P$ of $X$ and $x \in X$, we let $P(x)$ denote the unique
set in $P$ containing $X$.  We will require the following well-known
random partitioning lemma.

\begin{theorem}[\cite{CKR01}]
   \label{thm:ckr}
   For any finite metric space $(X,d)$ and
   value $\Delta > 0$, there is a random partition $P$ of $X$
   such that:
   \begin{enumerate}
      \item $\diam_X(S) \leq \Delta$ for every $S \in P$.
      \item For all $x,y \in X$, it holds that
         \[
            \Pr\left[P(x) \neq P(y)\right] \leq 8 \frac{d(x,y)}{\Delta} \log \frac{|B(x,\Delta)|}{|B(x,\Delta/8)|}.
         \]
   \end{enumerate}
\end{theorem}

For each $k \in \{0,1,\ldots,K\}$, let $P_k$ be a random partition of $X$ satisfying the conclusion of \pref{thm:ckr} with $\Delta = \tau^{-k}$.
Define a random map $\psi_k : X \to U_k$ as follows:
\[
   \psi_k(x) \seteq \f_k(B_X(P_k(x),\tau^{-k}/2)),
\]
where $\f_k$ is the map defined in \eqref{eq:selector}.

\begin{lemma}
   For every $x \in X$, it holds that $\llangle (\psi_0(x),0), (\psi_1(x),1), \ldots, (\psi_K(x),K) \rrangle$ is a path in $\dag$.
\end{lemma}

\begin{proof}
   It suffices to show that for any $k \in \{0,1,\ldots,K-1\}$, we have $(\psi_k(x), \psi_{k+1}(x)) \in A_k$.
   Define $u'=\psi_{k+1}(x)$ and $S \seteq B_X(P_k(x), \tau^{-k}/2)$.
   Then $\diam_X(S) \leq 2 \tau^{-k}$ and 
   \[
      d(x,u') = d(x,\psi_{k+1}(x)) \leq 2\tau^{-(k+1)} + \diam_X(B_X(P_{k+1}(x),\tau^{-(k+1)}/2)) \leq 4 \tau^{-(k+1)} < \tau^{-k}/2,
   \]
   where the last inequality follows from $\tau = 12$.  Hence $u' \in S \cap U_{k+1}$.  We can therefore apply \pref{lem:selector}
   to conclude that $(\psi_k(x),u') = (\f_k(S), u') \in A_k$, completing the proof.
\end{proof}

For $x \in X$, define $\Psi(x) \seteq \llangle (\psi_0(x),0), (\psi_1(x),1), \ldots, (\psi_K(x),K) \rrangle$.
From the preceding lemma, we know that $\Psi : X \to \cP_{\dag}$.

\begin{lemma}\label{lem:stretch}
   For any $x, y \in X$, it holds that
   \[
      \E\left[\dist_{\hat{\dag}}(\Psi(x),\Psi(y))\right] \leq O(\log n)\, d(x,y).
   \]
\end{lemma}

\begin{proof}
 From \pref{thm:ckr}, we have
   \begin{align*}
      \E\left[\dist_{\hat{\dag}}(\Psi(x),\Psi(y))\right] &\leq \sum_{k=0}^{K} \Pr[P_k(x) \neq P_k(y)] \cdot 10 \tau^{-k} \\
                                                         &\leq 80\, d(x,y) \sum_{k=0}^K \log \frac{|B(x,\tau^{-k})|}{|B(x,\tau^{-k}/8)|} \\
                                                         &\leq 80 \log(n)\,d(x,y),
   \end{align*}
   where in the last line we used $\tau = 12 \geq 8$.
\end{proof}

\begin{corollary}
\label{cor:lipschitz}
   It holds that $\hat{\dag}$ is $O(\log n)$-Lipschitz with respect to $(X,d)$.
\end{corollary}

\begin{proof}
   Consider any sequence $x_1,\ldots,x_m$, and let us map it to the random sequence $\Psi(x_1),\ldots,\Psi(x_m)$.
   Then from \pref{lem:stretch}, we conclude
   \[
      \sum_{j=1}^{m-1} \E\left[\dist_{\hat{\dag}}(\Psi(x_j), \Psi(x_{j+1}))\right] \leq O(\log n) \sum_{j=1}^{m-1} d(x_j,x_{j+1}).
   \]
   Hence there is a mapping $f : X \to \cP_{\dag}$ (that depends on the sequence $x_1,\ldots,x_m$)
   such that $\sum_{j=1}^{m-1} d(f(x_j),f(x_{j+1})) \leq O(\log n) \sum_{j=1}^{m-1} d(x_j,x_{j+1})$, completing the proof.
\end{proof}
\subsection{Compression}
\label{sec:compress}
	Let $\hat{\dag} = (\dag, \omega, \theta)$ be the $\tau$-geometric marked DAG constructed in \pref{sec:dag-construction}. 
	For a point $u \in V$, we let $\sigma(u)$ denote the number of  paths in $\dag$ that start at $u$ and end in a point of $X$.
\begin{observation}
\label{obs:sigma-induction}
	For $u \in V \setminus X$, it holds that 
	\begin{equation}
	\label{eq:sigma-induction}
			\sigma(u) = \sum_{v: uv \in A} \sigma(v).
	\end{equation}
\end{observation}
Say an edge $uv \in A$ is \emph{heavy} if $v \not\in X$ and $\sigma(v) > \sigma(u) / 2$; otherwise we say that \emph{$uv$ is light}. 
    Moreover, we say a path $\gamma = \langle u_1 u_2, u_2 u_3, \ldots, u_{m - 1} u_m \rangle$ in $\dag$ is \emph{heavy-light} if all the edges $u_1 u_2, u_2 u_3, \ldots, u_{m - 2} u_{m-1}$ are heavy and $u_{m-1} u_m$ is light.
The next lemma is straightforward and follows from \pref{eq:sigma-induction}.
    \begin{lemma}
    \label{lem:one-heavy-child}
    	For every $u \in V$, there is at most one heavy edge in $\dag$ leaving $u$.
    \end{lemma}
    Now we construct the marked DAG $\tilde{\dag} = (\dag', \omega', \theta')$ with $\dag' = (V, A')$ as follows. 
    We connect $u_i = (x_i, i) \in V$ to $u_j = (x_j, j) \in V$ for $1 \leq i < j \leq K$ in $\dag'$ if there is a heavy-light path $\gamma = \langle u_i u_{i+1}, u_{i+1}u_{i+2}, \ldots, u_{j-1} u_j \rangle$ from $u_i$ to $u_j$ in $\dag$. Note that by \pref{lem:one-heavy-child}, at most one such path can exist.
    We further set 
    \begin{align*}
       \omega'_{u_i u_j} &\seteq  10\tau^{-j + 1},  \\
       \theta'_{u_i u_j} &\seteq  \prod_{k=i}^{j-1} \theta_{u_k u_{k+1}}. 
 \end{align*}
%
    \begin{lemma}
    \label{lem:sigma-edge-bound}
    	For $uv \in A'$ with $v \notin X$ it holds that 
    	\[
    	\sigma(v) \leq \sigma(u) / 2.
    	\]
    \end{lemma}
    \begin{proof}
    	Since $uv \in A'$, there must be a heavy-light path $\gamma = \langle w_1 w_2, \ldots, w_{m-1} w_m \rangle$ in $\dag$ with $w_1 = u$ and $w_m = v$. 
    	 Clearly the values of $\sigma(\cdot)$ are non-increasing along the (directed) paths in $\dag$, hence we have
    	\[
    	\sigma(u) \geq \sigma(w_2) \geq \cdots \geq \sigma(w_{m-1}).
    	\]
    	Furthermore, as $w_{m-1} v$ is a light edge and $v \not\in X$, it follows that 
    	\[
    	\sigma(v) \leq \sigma(w_{m-1}) / 2 \leq \sigma(u) / 2,
    	\]
    	as desired.
    \end{proof}
    \begin{lemma}
    \label{lem:compression-combinatorial-depth}
	It holds that $\Delta_0(\dag') \leq O(\log |\mathcal{P}_{\dag}|)$.
\end{lemma}
\begin{proof}
	We will argue that for every $\gamma \in \cP_{\dag'}$, one has $|\gamma| = O(\log |\cP_\dag|)$.
   Let $\gamma = \langle u_1 u_2, \ldots, u_{m-1} u_m \rangle$. 
\pref{lem:sigma-edge-bound} implies that for $1 \leq i \leq m - 2$  we have $\sigma(u_{i+1}) \leq \sigma(u_i) / 2$.
    Further note that we have $\sigma(u_1) = |\cP_{\dag}|$, and also clearly  $\sigma(u_{m-1}) \geq 1$.
   Therefore,
   \[
   m - 2 \leq  \log_2(|\cP_{\dag}|),
   \]
   completing the proof.
\end{proof}

We now define the map $f: \cP_{\dag} \to \cP_{\dag'}$ as follows. For a path $\gamma \in \cP_\dag$, let $f(\gamma)$ denote the path obtained by contracting all the heavy edges in $\gamma$. More precisely, for $\gamma = \langle u_1 u_2, u_2 u_3, \ldots, u_{m - 1} u_m \rangle$, we define $f(\gamma)$ as follows.
Denote $i_0 \seteq 1$, and for $j = 1, 2, \ldots, m'$, let $i_j$ denote the $j$th index for which $u_{i_j - 1} u_{i_j}$ is a light edge. We then denote
\[
   f(\gamma) \seteq \langle u_{i_0} u_{i_1}, u_{i_1} u_{i_2}, \ldots, u_{i_{m' - 1}} u_{i_{m'}} \rangle.
\]



\begin{lemma}
	\label{lem:compression-information-depth}
	It holds that $\Delta_I(\tilde{D}) \leq \Delta_I(\hat{D})$.
\end{lemma}
\begin{proof}
As all the edges in $\dag'$ correspond to a path in $\dag$, $f$ is a surjective map.
 Furthermore, for $\gamma \in \cP_\dag$, one has $\theta(\gamma) = \theta'(f(\gamma))$, for $\theta(\cdot)$ defined as in \pref{eq:theta-prod-def} and $\theta'(\cdot)$ defined analogously, and thus we have
	\[
\Delta_I(\tilde{\dag}) = \max_{\gamma' \in \cP_{\dag'}} \log(1/\theta'(\gamma')) = \max_{\gamma \in \cP_{\dag}} \log(1/\theta'(f(\gamma))) \leq  \max_{\gamma \in \cP_{\dag}} \log(1/\theta(\gamma)) = \Delta_I(\hat{D}),
\]
completing the proof.
\end{proof}

\begin{lemma}
\label{lem:compression-distortion}
	For all $\gamma, \gamma' \in \cP_{\dag}$ it holds that
	\[
	 \dist_{\hat{\dag}}(\gamma, \gamma') = \dist_{\tilde{\dag}}(f(\gamma), f(\gamma')).
	 	\]
\end{lemma}
\begin{proof}
Denote $\gamma = \langle u_1 u_2, \ldots, u_{m-1} u_m \rangle$  and $\gamma' = \langle u'_1 u'_2, \ldots, u'_{m-1} u'_{m} \rangle$, and let $u_i = u'_i$ be the first vertex at which $\gamma$ and $\gamma'$ diverge so that we have
\[
\dist_{\hat{\dag}}(P, P') = \max(\omega_{u_i u_{i+1}}, \omega_{u_i u'_{i+1}}) = 10 \tau^{-i}.
\]
By \pref{lem:one-heavy-child},
at most one of $u_i u_{i+1}$ and $u_i u'_{i+1}$ can be heavy. Suppose that $u_i u_{i+1}$ is light. 
Take $j\seteq 1$ when $i=1$, and otherwise let $j \leq i$ be the maximum index for which $u_{j-1} u_j$ is light.
Further let $k \geq i$ be the minimum index for which $u'_k u'_{k+1}$ is a light edge. Note that $k$ is well-defined because  $u'_{m-1} u'_m$ is light. 
Now we have
\begin{align*}
   \dist_{\tilde{\dag}}(f(\gamma), f(\gamma')) &= \max(\omega'_{u_j u_{i+1}}, \omega'_{u_j u'_{k+1}}) \\ &= \max(\omega_{u_i u_{i+1}}, \omega_{u'_k u'_{k+1}}) 
	= \max( 10\tau^{-i}, 10\tau^{-k}) = 10\tau^{-i}, 
\end{align*}
as desired.
\end{proof}

\begin{lemma}
\label{lem:tilde-D-tau-geometric-marked}
    	The $\tilde{D}$ is a marked DAG that is also $\tau$-geometric.
\end{lemma}

    \begin{proof}
      We first establish the $\tau$-geometric property.
    	Consider $u, v, w \in V$ with $uv, vw \in A'$. Denote $v = (x, i)$ for some $1 \leq i \leq K$. Then by construction,
      we have $\omega'_{uv} \geq 10 \tau^{-i+1}$ and $\omega'_{vw} \leq 10 \tau^{-i}$, completing the proof.
    	
    	Next, we show that $\tilde{D}$ is a properly-constructed marked DAG.
      We need to establish that for $u \in V \setminus X$ it holds that 
	\begin{equation}
	\label{eq:Dhat-marked-DAG}
			\sum_{v: uv \in A'} \theta'_{uv} = 1.
	\end{equation}
	Let $v_0 \seteq u$ and let $\gamma = \langle u v_1, v_1 v_2, \ldots, v_{k-1} v_k \rangle$ be the maximal heavy path going out of $u$ for some $k \geq 0$, meaning that all the edges $v_i v_{i+1}$ are heavy for $0 \leq i \leq k - 1$. \pref{lem:one-heavy-child} implies that the choice of $\gamma$ is unique.

   Now using \pref{eq:Dhat-marked-DAG}, write
	\begin{align*}
\sum_{v\spc uv \in A'} \theta'_{uv} &=  \sum_{j=0}^{k-1} \sum_{\substack{y \neq v_{j+1} : \\ v_j y \in A}} \theta_{v_j y} \cdot  \prod_{\ell=0}^{j-1} \theta_{v_\ell v_{\ell + 1}} + \prod_{\ell=0}^{k-1} \theta_{v_\ell v_{\ell + 1}} \cdot  \left(\sum_{y\spc v_k y \in A} \theta_{v_k y} \right)   \\
                                    &=  \sum_{j=0}^{k-1} \sum_{\substack{y \neq v_{j+1} : \\ v_j y \in A}} \theta_{v_j y} \cdot  \prod_{\ell=0}^{j-1} \theta_{v_\ell v_{\ell + 1}} + \prod_{\ell=0}^{k-1} \theta_{v_\ell v_{\ell + 1}}  \\
                                    &= \sum_{j=0}^{k-2} \sum_{\substack{y \neq v_{j+1} : \\ v_j y \in A}} \theta_{v_j y} \cdot  \prod_{\ell=0}^{j-1} \theta_{v_\ell v_{\ell + 1}} +  \prod_{\ell=0}^{k-2} \theta_{v_\ell v_{\ell + 1}} \cdot \left(\sum_{y\spc v_{k-1} y \in A} \theta_{v_{k-1} y}\right)   \\
	 &\mathrel{\makebox[\widthof{=}]{\vdots}} \\
    &=  \sum_{\substack{y \neq v_{1} : \\ v_0 y \in A}} \theta_{v_0 y}  +  \theta_{v_0 v_{1}} \cdot \left(\sum_{y\spc v_{1} y \in A} \theta_{v_1 y}\right)  \\
	&= \sum_{y \spc v_0 y \in A} \theta_{v_0 y}  \\
	&= 1,
	\end{align*} 
	as desired.
    \end{proof}


We are now ready to prove \pref{thm:main-2}.
\begin{proof}[Proof of \pref{thm:main-2}]
	Let us show that $\tilde{\dag}$ satisfies the requirements of the theorem. \pref{lem:tilde-D-tau-geometric-marked} shows that $\tilde{\dag}$ is a $12$-geometric marked DAG.

   Now note that for $\gamma \in \cP_\dag$ we have $\overline{f(\gamma)} = \bar{\gamma}$, and thus \pref{lem:compression-distortion} in conjunction with \pref{lem:expanding} and \pref{cor:lipschitz} implies that $\tilde{D}$ is 
	$1$-expanding and $O(\log n)$-Lipschitz.
	Moreover, \pref{lem:compression-information-depth} together with \pref{lem:root-leaf-sum-bound} bounds the information depth of $\tilde{D}$.
   Finally, a bound on the combinatorial depth follows from \pref{lem:compression-combinatorial-depth} and \pref{cor:root-leaf-path-count}.
\end{proof}

\section{Algorithm and competitive analysis}
\label{sec:discrete-time}

\subsection{Discrete-time algorithm}
\label{sec:discrete-time-alg}
Let $\hat{\dag} = (\dag, w, \theta)$ be a marked DAG on $X$ with $\dag = (V, A)$.
 We now describe a generalization of the discrete-time dynamics of \cite{CL19} on $\hat{\dag}$ in response to a sequence of costs $\langle c_t : t \geq 1\rangle$, where $c_t \in \R_+^{X}$. 
 Define
\[
\cQ_\dag \seteq \left\{p \in \R_+^{A} \bigmid \forall u \in V \setminus X: \sum_{v \spc uv \in A} p_{uv} = 1 \right\}.
\]
For $q \in Q_{\dag}$ and $u \in V \setminus X$, we use $q^{(u)}$ to denote the restriction of $q$ to the subspace spanned by subset of
standard basis vectors $\{e_{uv}: uv \in A\}$, and we define
the corresponding probability simplex
$Q^{(u)}_{\dag} \seteq \{ q^{(u)} : q \in Q_{\dag} \}$. For convenience, we use $q^{(u)}_v$ for
$q^{(u)}_{uv}$.

Let $\kappa > 0$ be a normalization parameter, and let the values $\eta_{uv}$ and $\delta_{uv}$ be defined as in \pref{eq:eta-def} and \pref{eq:delta-def}.
For $u \in V \setminus X$ and $p \in \cQ_\dag^{(u)}$, define
\[
\Phi^{(u)}(p) \seteq \frac{1}{\kappa}  \sum_{v \spc uv \in A} \frac{\omega_{uv}}{\eta_{uv}}(p_{uv} + \delta_{uv}) \log(p_{uv} + \delta_{uv}),
\]
and for $p' \in \cQ_\dag^{(u)}$, denote
\[
\vvD^{(u)}\!\left(p \dmid p' \right) \seteq \vvD_{\Phi^{(u)}}\!\left(p \dmid p' \right) = \frac{1}{\kappa} \sum_{v \spc uv \in A} \frac{\omega_{uv}}{\eta_{uv}} \left[\left(p_{uv} + \delta_{uv}\right) \log \frac{p_{uv}+\delta_{uv}}{p'_{uv}+\delta_{uv}}
+ p'_{uv}-p_{uv}\right].
\]


We now define an algorithm that takes a point $q \in Q_{\dag}$ and a cost vector $c \in \R_+^{X}$ and outputs
a point $p = \cA(q,c)\in Q_{\dag}$.
Fix a topological ordering $\langle u_1,u_2,\ldots, u_N\rangle$ of $V \setminus X$ in $\dag$.
We define $p$ inductively as follows.
Denote $\hat{c}_x \seteq c_{x}$ for $x \in X$, and
for each $j=1,2,\ldots,N$:
\begin{align}
   \hat{c}^{(u_j)}_{v} &\seteq \hat{c}_v \qquad \forall v: (u_j, v) \in A \label{eq:hatc}\\
p^{(u_j)} &\seteq \argmin \left\{ \vvD^{(u_j)}\!\left(p \dmid q^{(u_j)}\right) + \llangle p, \hat{c}^{(u_j)} \rrangle \bigmid 
p \in Q^{(u_j)}_\dag \right\} \label{eq:locp} \\
\hat{c}_{u_j} &\seteq \sum_{v \spc u_j v \in A} p^{(u_j)}_{v}\,\hat{c}_v \label{eq:hatc2}
\end{align}

We will use $\flow^{\dag} : Q_{\dag} \to \cF_\dag$ for the map which sends $q \in Q_{\dag}$ to the (unique) $F = \flow^{\cD}(q) \in  \cF_\dag$ such that
\[
   F_{uv} = F_u q_{uv} \qquad \forall uv \in A.
\]
Note that $q$ contains more information than $F$; the map $\flow^{\dag}$ fails to be invertible at $F \in \cF_{\dag}$ whenever 
there is some $u \in V \setminus X$ with $F_u = 0$.
We will drop the superscript $\dag$ from $\flow^{\dag}$ whenever it is clear from context.


Now let $p_0$ be an arbitrary point in $Q_{\dag}$. Given the cost sequence $\llangle c_t : t \geq 1\rrangle$, for $t = 1, 2, \ldots$ we define
\begin{equation}
\label{eq:inductive-p}
	p_t \seteq \cA(p_{t-1}, c_t),
\end{equation}
and the associated MTS algorithm plays the distribution $\flow^{\dag}(p_t)|_X$, i.e., at every $x \in X$ (recall that
these are precisely the sinks in $\dag$), the algorithm places probability mass equal to the flow in $\flow^{\dag}(p_t)$ entering $x$.

For $c \in \R_+^{X}$ and $F \in \cF_\dag$ we define
	\[
      \llangle c, F \rrangle_X \seteq \sum_{v \in X} c_v F_v = \sum_{uv \in A \spc v \in X} c_v F_{uv}.
	\]
So the service cost of the algorithm until time $t \geq 1$ is given by
\[
   \sum_{s=1}^{t} \llangle c, \flow^{\dag}(p_s) \rrangle_X,
\]
and the movement cost is given by
\[
   \sum_{s=1}^t \vvW^1_{\hat{D}}\!\left(\flow^{\dag}(p_{s-1}),\flow^{\dag}(p_{s})\right),
\]
where we recall the $L^1$ transportation distance defined in \pref{sec:metric-compat}.

\subsection{Analysis via unfolding to an ultrametric}
Let $\hat{\dag} = (\dag, \omega, \theta)$ be a $\tau$-geometric marked DAG. As in \pref{sec:compress}, for a point $u \in V$, we define $\sigma(u)$ to denote the number of  paths in $\dag$ that start at $u$ and end at $X$.
Then if $\dag$  is a tree and furthermore, for $uv \in A$, one defines
\begin{equation}
\label{eq:theta-equals-sigma}
	\theta_{uv} \seteq \frac{\sigma(v)}{\sigma(u)},
\end{equation}
 then the algorithm of the preceding section is
 exactly the same as the one for HSTs introduced in \cite{CL19},
 as $\sigma(u)$ is precisely the number of leaves in the subtree rooted at $u$.
The next result is a restatement of \cite[Thm. 2.7]{CL19}.

\begin{theorem}[\cite{CL19}]
\label{thm:cl19}
Let $\hat{\dag}=(\dag,\omega,\theta)$ be a $\tau$-geometric marked DAG over $X$,
   for some $\tau \geq 4$, and such that $\dag$ is a tree.
   If $\theta$ is defined as in \pref{eq:theta-equals-sigma}, and
   $\hat{\dag}$ is $1$-expanding and $L$-Lipschitz,
   then there is some value $\kappa \asymp L$ and a number $\e=\e({\hat{\dag}}) > 0$ so that for any
   sequence of cost vectors $\langle c_t : t \geq 1\rangle$ satisfying $\|c_t\|_{\infty} \leq \e$,
   the MTS algorithm specified in \pref{sec:discrete-time-alg} is $1$-competitive for service costs and
   $O\!\left(L \left(\Delta_0(\dag) + \log |X|)\right)\right)$-competitive for movement costs.
\end{theorem}
Note that the condition on the $\ell_{\infty}$ norm
of the cost vectors in the above theorem is not restrictive, since as noted in \cite{CL19}, we can always split arbitrary cost vectors into smaller pieces with each satisfying the desired $\ell_\infty$ bound.

Our goal now is to show that if $\theta$ is defined as in \pref{eq:theta-equals-sigma}, then
similar guarantees as in \pref{thm:cl19} hold for the algorithm on $\hat{\dag}$, even when $\dag$ is not a tree.

\begin{theorem}
\label{thm:unfolded-DAG}
   Let $\hat{\dag}$ be a $\tau$-geometric marked DAG over $X$,
   for some $\tau \geq 4$, and such that $\theta$ is given by \pref{eq:theta-equals-sigma}.
   If $\hat{\dag}$ is $1$-expanding and $L$-Lipschitz,
   then there is some value $\kappa \asymp L$ and a number $\e=\e({\hat{\dag}}) > 0$ so that for any
   sequence of cost vectors $\langle c_t : t \geq 1\rangle$ satisfying $\|c_t\|_{\infty} \leq \e$,
   the MTS algorithm specified in \pref{sec:discrete-time-alg} is $1$-competitive for service costs and
   $O\!\left(L \left(\Delta_0(\dag) + \log |X|)\right)\right)$-competitive for movement costs.
\end{theorem}

	Note that from \pref{eq:Delta-I-PD-bound} it follows that 
	\[
	\Delta_0(\dag) + \log |\cP_\dag| \leq \Delta_0(\dag) + \Delta_I(\hat{\dag}),
	\]
	and hence the above theorem together with \pref{thm:main-2} already
   gives a competitive algorithm with our desired bounds,
   though only for the specific choice of $\theta$ given by \pref{eq:theta-equals-sigma}.
   In \pref{sec:analysis-general-case}, we address the case of general $\theta$.
	
We prove \pref{thm:unfolded-DAG} via a simple reduction to \pref{thm:cl19}. 
Consider a $\tau$-geometric marked DAG $\hat{\dag} = (\dag, w, \theta)$  on $X$ with $\dag = (V, A)$. Note that $d_{\hat{\dag}}$ defines an ultrametric on $\cP_\dag$. We show that the dynamics on $\hat{\dag}$ are ``equivalent'' to the dynamics on the HST corresponding to the ultrametric $(\cP_\dag, d_{\hat{\dag}})$. More precisely, let us construct the $\tau$-geometric marked tree
 $\tilde{\dag} = (\dag', w', \theta')$ 
 with $\dag' = (V', A')$  as follows. We define $V'$ as the set of
 (directed) paths is $\dag$ originating from the root. 
 Furthermore, we  connect $\gamma \in V'$ to $\gamma' \in V'$ whenever $\gamma'$ is formed by adding the edge $\bar{\gamma'}$ to $\gamma$, and set
\[
\omega_{\gamma \gamma'} \seteq \omega_{\bar{\gamma'}}, \qquad \theta_{\gamma \gamma'} \seteq \theta_{\bar{\gamma'}}.
\]
One can verify that $\tilde{\dag}$ is a $\tau$-geometric marked tree over $\cP_\dag$. 
Moreover, since $\dag'$ is a tree, there is a natural identification between the elements of $\cP_\dag$ and $\cP_{\dag'}$ so that for $\gamma, \gamma' \in \cP_\dag$ it holds that
\begin{equation}
\label{eq:Dhat-Dtilde-isometric}
	d_{\hat{\dag}}(\gamma, \gamma') = d_{\tilde{\dag}}(\gamma, \gamma').
\end{equation}
Now for $p \in \cQ_\dag$, define $\tilde{p} \in Q_{\dag'}$ to be the natural extension of $p$ in $\dag'$ so that for $\gamma \gamma' \in A'$
one has $\tilde{p}_{\gamma \gamma'} = p_{\bar{\gamma'}}$. 
Furthermore, for a cost sequence $c \in \R_+^{X}$  define its extension $\tilde{c} \in \R^{\cP_\dag}$ as the vector with $\tilde{c}_\gamma = c_{\bar{\gamma}}$ for $\gamma \in \cP_\dag$. Finally, let $\cA$ denote the single-step discrete dynamics on $\hat{\dag}$ as defined in \pref{sec:discrete-time-alg}, and similarly let $\cA'$ denote the discrete dynamics on $\tilde{\dag}$.  Then the following lemma is straightforward.
\begin{lemma}
\label{lem:equivalent-dynamics}
	Let $p \in \cQ_\dag$, $c \in \R_+^X$. Then it holds that 
	\begin{equation}
\label{eq:service-cost-equiv}
	\langle \flow^{(\dag)}(p), c \rangle_X = \langle \flow^{(\dag')}(\tilde{p}), \tilde{c} \rangle_{\cP_\dag}.
\end{equation}
Furthermore, for $q = \cA(p, c)$ we have
	\begin{equation}
	\label{eq:Apq-equiv}
		\cA'(\tilde{p}, \tilde{c}) = \tilde{q}.
	\end{equation}

\end{lemma}


We are now ready to prove the main result of this section.
\begin{proof}[Proof of \pref{thm:unfolded-DAG}]
Let $p_0 \in \cQ_\dag$ and $q_0 \in Q_{\dag'}$ with $q_0 = \tilde{p_0}$. Given the cost sequence $\llangle c_t : t \geq 1\rrangle$, for $t \geq 1$  let 
\[
p_t = \cA(p_{t-1}, c_t)
\]
and
\[
q_t = \cA'(q_{t-1}, \tilde{c}_t).
\]
Then by repeatedly applying \pref{eq:Apq-equiv} we get that for $t \geq 1$ we have $q_t = \tilde{p_t}$.
Therefore from \pref{eq:service-cost-equiv} and \pref{eq:Dhat-Dtilde-isometric} it follows that the service and movement costs of the dynamics on $\hat{\dag}$ and $\tilde{\dag}$ are equal.
Hence the competitiveness guarantees for the dynamics on $\hat{\dag}$ follow from an application of \pref{thm:cl19} to the dynamics on $\tilde{\dag}$, completing the proof.
\end{proof}

\subsection{Analysis of the general case}
\label{sec:analysis-general-case}

We now prove \pref{thm:main-1} via a relatively straightforward generalization of the analysis in \cite{CL19}.
Let $\hat{\dag} = (\dag, w, \theta)$ be a $\tau$-geometric marked DAG with $\tau \geq 4$, and consider the mirror descent dynamics on $\hat{\dag}$ described in \pref{sec:discrete-time-alg}.

For a unit flow $F \in \cF_\dag$
 and $q \in \cQ_\dag$, 
define the global divergence function
\[
\vvD(F \dmid q) \seteq   \frac{1}{\kappa} \sum_{uv \in A} \frac{\omega_{uv}}{\eta_{uv}} \left[ \left(F_{uv} + F_u \delta_{uv}\right) \log \frac{\frac{F_{uv}}{F_u}+\delta_{uv}}{q_{uv}+\delta_{uv}}
+ F_u q_{uv}-F_{uv}\right],
\]
with the convention that $0 \log \left(\frac{0}{0}+\delta_v\right) = \lim_{\eps \to 0} \eps \log \left(\frac{0}{\eps}+\delta_v\right) = 0$.
We further define the norm $\ell_1(\omega)$ as 
\[
\|F\|_{\ell_1(\omega)} = \sum_{uv \in A} \omega_{uv} |F_{uv}|.
\]
\begin{observation}
\label{obs:wasserstein-l1-equiv}
	For $F, F' \in \cF_\dag$ it holds that
\[
\frac{1}{2}\|F - F'\|_{\ell_1(\omega)} \leq \vvW_{\hat{\dag}}^1(F, F') \leq \|F - F'\|_{\ell_1(\omega)}.
 \]
\end{observation}
The next lemma lets us bound the amount of change of the global divergence when the offline algorithm makes a movement.
\begin{lemma}[{\cite[Lemma~2.2]{CL19}}]
\label{lem:lipschitz}
	For flows $F, F' \in \cF_\dag$ and $q \in \cQ_\dag$ we have
	\[
	|\vvD(F \dmid q) - \vvD(F' \dmid q)| \leq \frac{1}{\kappa} (2 + \frac{4}{\tau}) \|F- F'\|_{\ell_1(\omega)}
	\]
\end{lemma}

Suppose  $q \in \cQ_\dag$, $p = \cA(q, c)$, and further let $Q = \mu(q), P = \mu(p)$. 
The KKT conditions  for \pref{eq:locp} give:  For every $uv \in A$,
\begin{equation}\label{eq:kkt2}
\frac{1}{\kappa} \frac{\omega_{uv}}{\eta_{uv}} \log \left(\frac{p_{uv}+\delta_{uv}}{q_{uv}+\delta_{uv}}\right) = \beta_u - \hat{c}_{v} + \alpha_{uv}\,,
\end{equation}
where $\alpha_{uv}$ is the Lagrange multipliers corresponding to the nonnegativity constraints in \pref{eq:locp}, $\beta_{u} \geq 0$ is the multiplier corresponding to the constraint $\sum_{v: uv \in A} q_{uv} \geq 1$, and $\hat{c}$ is defined as in \pref{eq:hatc2}. 
Note that as in \cite{CL19} the nonnegativity multipliers are unique
and thus well-defined here.
The complementary slackness conditions give us
\begin{equation}\label{eq:posmult}
\alpha_{uv} > 0 \implies p_{uv} = 0.
\end{equation}
We use $\alpha^{(u)}$ to denote the restriction of $\alpha$ to the subspace spanned by $\{e_{uv}: uv \in A\}$.


The following two lemmas, which allow us to bound the service cost and the movement cost of the algorithm, respectively, are the main ingredients in the proof of \pref{thm:main-2}.
\begin{lemma}\label{lem:sc}
	 It holds that
	\[
	\vvD(F \dmid p) - \vvD(F \dmid q) \leq \llangle c, F - P \rrangle_{X}.
	\] 
\end{lemma}
Define $\omega_{\min} \seteq \min_{uv \in A} \{ \omega_{uv} \}$
and
\[
\eps_\dag \seteq \frac{\omega_{\min}}{2 (2\Delta_0(\dag)+\Delta_I(\hat{\dag}))} \frac{\tau-3}{\tau\kappa}.
\]
Furthermore, for $F \in \cP_\dag$ and $r \in \cQ_\dag$ define
\begin{align*}
	\psi(F) \seteq \sum_{uv \in A} \omega_{uv} F_{uv}
\end{align*}
and
\begin{align*}
   \Psi_u(r) &\seteq - \flow^{\dag}(r)_u\,\vvD^{(u)}\!\left(\theta^{(u)} \dmid r^{(u)}\right) \\
\Psi(r) &\seteq \sum_{u \in V \setminus X} \Psi_u(r).
\end{align*}
\begin{lemma}\label{lem:movement-analysis}
	For any $Z \in \cF_\dag$:
	\begin{align}
	\kappa^{-1} \left\|Q-P\right\|_{\ell_1(\omega)} &\leq [\psi(Y)-\psi(X)] + \frac{2 \tau}{\tau-3} \left([\Psi(q)-\Psi(p)] + (2\Delta_0(\dag) + \Delta_I(\hat{\dag})) \langle c, Q\rangle_{X}\right).\label{eq:mvmtx}
	\end{align}
	Moreover, if $\|c\|_{\infty} \leq \eps_\dag$, then
	\begin{equation}\label{eq:mvmty}
	\kappa^{-1} \left\|Q-P\right\|_{\ell_1(\omega)} \leq [\psi(Y)-\psi(X)] + \frac{4 \tau}{\tau-3} \left( [\Psi(q)-\Psi(p)] + (2\Delta_0(\dag) + \Delta_I(\hat{\dag})) \langle c, P\rangle_{X}\right).
	\end{equation}
\end{lemma}
We prove \pref{lem:sc} and \pref{lem:movement-analysis} in \pref{sec:service-cost} and \pref{sec:movement-cost}, respectively. Now given these results, let us prove \pref{thm:main-2}.
\begin{proof}[Proof of \pref{thm:main-2}]
	Consider a sequence $\llangle c_t : t \geq 1\rrangle$ of cost functions.
	By splitting the costs into smaller pieces, we may assume that $\|c_t\|_{\infty} \leq \eps_\dag$ for all $t \geq 1$.
	
	Let $t_1 \geq 1$, and let $r^*_0, r^*_1, \ldots, r^*_{t_1} \in X$ denote the path taken by an (optimal) offline algorithm in response to the cost sequence $\llangle c_t : t \geq 1\rrangle$. The $L$-Lipschitzness property of $\hat{\dag}$ implies that there exists a sequence $R^*_0, R^*_1, \ldots R^*_{t_1} \in \cF_\dag$  such that $R^*_i$ is a unit flow to $r_i$, and furthermore
	\begin{equation}
	\label{eq:lip-z}
		\sum_{i=1}^{t_1} \vvW_{\hat{\dag}}^1(R^*_{i-1}, R^*_i) \leq L \sum_{i=1}^{t_1} d(r^*_{i-1},r^*_i).
	\end{equation}
	
Let $q_0, \ldots, q_{t_1} \in \cQ_\dag$ denote the trajectory of the discrete mirror descent dynamics with $\kappa = 6L$ on $\hat{\dag}$ in response to the cost sequence $\llangle c_t : t \geq 1\rrangle$.
 Further let $\{Q_t = \mu(q_t)\}$, and suppose $R_0^* = Q_0^*$.
	Then using $\vvD(R_{0}^* \dmid q_0) = 0$ along with \pref{lem:sc} and \pref{lem:lipschitz}
	yields, for any time $t_1 \geq 1$,
	\begin{align*}
	\sum_{t=1}^{t_1} \langle c_t, Q_t\rangle_{X} &\leq 
	\sum_{t=1}^{t_1} \langle c_t,Z_t^*\rangle_{X} -
	\vvD(R_{t_1}^* \dmid q_{t_1}) + \frac{3}{\kappa} \sum_{t=1}^{t_1} \|R_t^*-R_{t-1}^*\|_{\ell_1(\omega)} \\
	&\leq \sum_{t=1}^{t_1} \langle c_t,R_t^*\rangle_{X}
	+ \frac{3}{\kappa} \sum_{t=1}^{t_1} \|R_t^*-R_{t-1}^*\|_{\ell_1(\omega)} \\
	&\leq \sum_{t=1}^{t_1} \langle c_t,R_t^*\rangle_{X}
	+ \frac{6L}{\kappa} \sum_{t=1}^{t_1} d(r^*_{t-1}, r^*_t),
	\end{align*}
	where in the second line we have used $\vvD(R \dmid q) \geq 0$ for all $R \in \cF_\dag$ and $q \in \cQ_\dag$, and the last line follows from \pref{obs:wasserstein-l1-equiv} and \pref{eq:lip-z}. This confirms that the mirror descent dynamics is $1$-competitive for the service costs.
Now we can write
	\begin{align*}
	\frac{\epsilon}{\kappa} \sum_{t=1}^{t_1} \vvW^1_X(Q_{t-1}, Q_t) &\leq
		\frac{1}{\kappa} \sum_{t=1}^{t_1} \vvW^1_{\hat{\dag}}(Q_{t-1}, Q_t) \tag{$\hat{\dag}$ is $\epsilon$-expanding}\\ 
	 &\leq
		\frac{1}{\epsilon} \sum_{t=1}^{t_1} \|Q_t-Q_{t-1}\|_{\ell_1(\omega)} \tag{\pref{obs:wasserstein-l1-equiv}} \\ 
	&\leq  \left[\psi(Q_{t_1})-\psi(Q_0)\right]+\frac{4\tau}{\tau-3}  \left( \left[\Psi(q_0)-\Psi(q_{t_1})\right] + \left(2\Delta_0(\dag) + \Delta_I(\hat{\dag})\right) 
	\sum_{t=1}^{t_1} \langle c_t,Q_t\rangle_{X}\right), 
	\end{align*}
where in the last line we used \pref{eq:mvmty}.
	This implies that the mirror descent dynamics is $(96L/\epsilon) \cdot \left(2\Delta_0(\dag) + \Delta_I(\hat{\dag})\right)$-competitive in the movement cost, completing the proof.
\end{proof}

\subsection{Bounding the service cost}
\label{sec:service-cost}
In this section we prove \pref{lem:sc}.
Let $F \in \cF_\dag$, and for $u \in V \setminus X$ with $F_u > 0$, define $F^{(u)} \in Q^{(u)}_\dag$ by
	\[
	F^{(u)}_v \seteq \frac{F_{uv}}{F_u}. 
	\]
The next lemma is a consequence of \cite[Lemma~2.1]{CL19}.
\begin{lemma}
For $u \in V \setminus X$ we have
	\begin{equation}
	\label{eq:bregman-improvement-u}
		\vvD^{(u)}\left(F^{(u)} \dmid p^{(u)}\right) - \vvD^{(u)}\left(F^{(u)} \dmid q^{(u)}\right) \leq \llangle \hat{c}^{(u)} - \alpha^{(u)}, 
	F^{(u)} - p^{(u)}\rrangle.	
	\end{equation}
\end{lemma}
	

\begin{proof}[Proof of \pref{lem:sc}]
	Multiplying both sides of \pref{eq:bregman-improvement-u} by $F_u$ and summing over all $u \in V \setminus X$ yields
	\begin{align*} 
	\vvD(F \dmid p) - \vvD(F \dmid q) &\leq \sum_{u \in V \setminus X} F_u \llangle \hat{c}^{(u)} - \alpha^{(u)}, F^{(u)} - p^{(u)} \rrangle \\
	&=\sum_{uv \in A} F_{uv} (\hat{c}^{(u)}_v -\alpha^{(u)}_v) - \sum_{uv \in A} F_u p_{uv} (\hat{c}^{(u)}_v-\alpha^{(u)}_v)  \\
	&\leq \sum_{uv \in A} F_{uv} \hat{c}^{(u)}_v  - \sum_{uv \in A} F_u p_{uv} (\hat{c}^{(u)}_v-\alpha^{(u)}_v) \tag{$\alpha^{(u)}_v \geq 0$}. `\,
	\end{align*}
	Note that from  \pref{eq:posmult}
	the latter expression is
	\[
	\sum_{u \notin X} F_u \sum_{v: uv \in A} \hat{c}_v^{(u)} p_v
	\stackrel{}{=}\sum_{u \notin X} F_u \hat{c}_u.
	\]
	Noting that $\hat{c}_{\rt} = \sum_{u \in X} \mu(p)_{u} c_{u}$,  this gives
	\[
	\vvD(F \dmid p) - \vvD(F \dmid q) \leq \sum_{u \neq \rt} \hat{c}_u F_u - \sum_{u \in V \setminus X} F_u \hat{c}_u
	\leq \llangle c, F-P\rrangle_{X}. \qedhere
	\]
\end{proof}
\subsection{Bounding the the movement cost}
\label{sec:movement-cost}
In this section we prove \pref{lem:movement-analysis}.
The next lemma shows that when the algorithm moves from $Q$ to $P$ it suffices for us to bound the positive movement movement cost $\left\|(P-Q)_{+}\right\|_{\ell_1(\omega)}$.
\begin{lemma}[{\cite[Lemma~2.4]{CL19}}]
	\label{lem:height}
	For $F,F'\in \cF_\dag$ it holds that
	\[
	\|F-F'\|_{\ell_1(\omega)} = 2 \left\|(F-F')_{+}\right\|_{\ell_1(\omega)} + [\psi(F')-\psi(F)].
	\]
\end{lemma}

\begin{lemma}[{\cite[Lemma~2.9]{CL19}}]
	\label{lem:alphas}
	It holds that $\alpha_{uv} \leq \hat{c}_{v}$ for all $uv \in A$.
\end{lemma}

Define
$\rho_{uv} \seteq \log \left(\frac{p_{uv}+\delta_{uv}}{q_{uv}+\delta_{uv}}\right)$
so that
\begin{align}
q_{uv} - p_{uv} &= (q_{uv} + \delta_{uv}) (1-e^{\rho_{uv}}). \label{eq:qp}
\end{align}
Recall that for $uv \in A$, we have $Q_{uv} = q_{uv} Q_u$ and $P_{uv} = p_{uv} P_u$, thus
\[
Q_{uv} - P_{uv} = Q_u (q_{uv} - p_{uv}) + p_{uv} (Q_u - P_u) = (Q_{uv}+\delta_{uv} Q_u) (1-e^{\rho_{uv}}) + p_{uv} (Q_u - P_u).
\]
In particular,
\begin{align*}
\omega_{uv} \left(Q_{uv}-P_{uv}\right)_+ &\leq \omega_{uv} (Q_{uv}+\delta_{uv} Q_u) (1-e^{\rho_{uv}})_+ + \omega_{uv} p_{uv} \left(Q_u-P_u\right)_+  \\
&\leq \omega_{uv} (Q_{uv}+\delta_{uv} Q_u) (1-e^{\rho_{uv}})_+ + \sum_{w: wu \in A} \omega_{uv} p_{uv} \left(Q_{wu}-P_{wu}\right)_+ \\
&\leq \omega_{uv} (Q_{uv}+\delta_{uv} Q_u) (1-e^{\rho_{uv}})_+ + \sum_{w: wu \in A} \frac{\omega_{wu}}{\tau} p_{uv} \left(Q_{wu}-P_{wu}\right)_+.
\end{align*}
Using $\sum_{v: uv \in A} p_{uv} = 1$ and
summing over all edges yields
\[
\sum_{uv \in A} \omega_{uv} \left(Q_{uv}-P_{uv}\right)_+ \leq\sum_{uv \in A} \omega_{uv}  (Q_{uv}+\delta_{uv} Q_{u}) (1-e^{\rho_{uv}})_{+} + \frac{1}{\tau} \sum_{uv \in A} \omega_{uv} \left(Q_{uv}-P_{uv}\right)_+\,,
\]
hence
\begin{align}
\sum_{uv \in A} \omega_{uv} \left(Q_{uv}-P_{uv}\right)_+ &\leq \frac{\tau}{\tau-1} \sum_{uv \in A} \omega_{uv}  (Q_{uv}+\delta_{uv} Q_{u}) (1-e^{\rho_{uv}})_{+} \nonumber \\ 
&\leq \frac{\tau}{\tau-1}  \sum_{uv \in A} \omega_{uv}  (Q_{uv}+\delta_{uv} Q_{u}) (\rho_{uv})_{-} \nonumber \\
&\leq\frac{\kappa \tau}{\tau-1} \left(\sum_{uv \in A} \eta_{uv} Q_{uv} \hat{c}_v + \sum_{uv \in A} Q_u \theta_{uv} (\hat{c}_v-\alpha_{uv})\right),
\label{eq:mvmt0}
\end{align}
where the last line uses \pref{lem:alphas} and \pref{eq:kkt2}, to bound
$\omega_{uv} (\rho_{uv})_{-} \leq \kappa \eta_{uv}\left(\hat{c}_v - \alpha_{uv}\right)$. 
\begin{lemma}
\label{lem:mv2}
	It holds that 
	\[
\sum_{uv \in A} \eta_{uv} Q_{uv} \hat{c}_v
\leq (\Delta_0(\dag) + 
\Delta_I(\hat{\dag})) \llangle c,Q\rrangle_X.
\]
\end{lemma}
\begin{proof}
	Consider a decomposition of $Q$ into flows on single source-sink paths. More precisely, let $\chi: \cP_\dag \to \R_+$ be so that 
\[
Q = \sum_{\gamma \in \cP_\dag} \chi(\gamma) \mathbb{1}_\gamma.
\]
Note that the existence of such a decomposition is guaranteed by \pref{eq:flow-def}. Now we have
\begin{align*} 
\sum_{uv \in A} \eta_{uv} Q_{uv} \hat{c}_v \leq 
 \sum_{\gamma \in \cP_\dag} c_{\bar{\gamma}} Q_{\bar{\gamma}} \chi(\gamma) \sum_{uv \in \gamma} \eta_{uv}
\leq (\Delta_0(\dag) + 
\Delta_I(\hat{\dag})) \llangle c,Q\rrangle_X,
\end{align*}
since for any $\gamma \in \cP_\dag$, we have
\[
\sum_{uv \in \gamma} \eta_{uv} = |\gamma| + \log(1/\theta(\gamma))  \leq \Delta_0(\dag)  + \Delta_I(\hat{\dag}).\qedhere
\]
\end{proof}



It only remains to bound the latter term in \pref{eq:mvmt0}.
In order to do so, we would need the following result from \cite{CL19}.

\begin{lemma}[{\cite[Lemma~2.11]{CL19}}]
   \label{lem:crucial}
	For any $u \in V \setminus X$, it holds that
	\begin{equation}\label{eq:hybrid}
	\Psi_u(p) - \Psi_u(q) \leq \frac{2}{\kappa}  \left(Q_u-P_u\right)_+ \cdot \max_{v: uv \in A} \omega_{uv} + 
	\sum_{v: uv \in A} (\hat{c}_v-\alpha_{uv}) \left[Q_{uv} - \theta_{uv} Q_u\right].\qedhere
	\end{equation}
\end{lemma}

We omit a proof of the lemma as it is essentially identical to that of \cite[Lem.~2.11]{CL19}. 
In \cite{CL19}, for a fixed $u$, the probability distirbution specified by
$\langle \theta_{uv} : uv \in A \rangle$  is uniform, but the argument works verbatim for any probability.

\begin{lemma}
	\label{lem:movement}
   It holds that
	\[
	\frac{\tau-3}{\kappa \tau} \left\|\left(Q-P\right)_+\right\|_{\ell_1(\omega)} \leq (2 \Delta_0(\dag) + \Delta_I(\hat{\dag})) \langle c, Q\rangle_{X} +
	\left[\Psi(q)-\Psi(p)\right].
	\]
\end{lemma}
\begin{proof}
	Using \pref{lem:crucial} gives
\begin{align*}
 \sum_{uv \in A} Q_u \theta_{uv} (\hat{c}_v-\alpha_{uv}) 
&\stackrel{\mathclap{\pref{eq:hybrid}}}{\leq}\ 
[\Psi(q)-\Psi(p)] + \frac{2}{\kappa \tau} \left\|\left(Q-P\right)_{+}\right\|_{\ell_1(\omega)}
+ \sum_{uv \in A} Q_{uv} \hat{c}_v \\
&\leq\ 
[\Psi(q)-\Psi(p)] + \frac{2}{\kappa \tau} \left\|\left(Q-P\right)_{+}\right\|_{\ell_1(\omega)}
+ \Delta_0(\dag) \langle c,Q\rangle_{X}.
\end{align*}
Combining this inequality with \pref{eq:mvmt0} and \pref{lem:mv2} gives
\begin{equation*}
\kappa^{-1} \left\|\left(Q-P\right)_+\right\|_{\ell_1(\omega)} 
\leq \frac{\tau}{\tau-1}\left[
\left(2\Delta_0(\dag) + \Delta_I(\hat{\dag})\right) \langle c,Q\rangle_{X} + \left(\Psi(q)-\Psi(p)\right) 
+ \frac{2}{\kappa\tau} \left\|\left(Q-P\right)_+\right\|_{\ell_1(\omega)}\right],
\end{equation*}
completing the proof.
\end{proof}

\begin{proof}[Proof of \pref{lem:movement-analysis}]
	\pref{eq:mvmtx}
	follows from \pref{lem:movement} and \pref{lem:height}.
	To see that \pref{eq:mvmty} follows from \pref{eq:mvmtx} and \pref{lem:movement}, use the fact that
	\[
	\langle c,Q \rangle_{X} \leq \langle c,P\rangle_{X} + \frac{\|c\|_{\infty}}{\omega_{\min}} \left\|\left(Q-P\right)_+\right\|_{\ell_1(\omega)}.\qedhere
	\]
\end{proof}



\label{sec:movement-analysis}

\subsection*{Acknowledgements}

We thank the anonymous referees for their helpful comments. This research was partially supported by NSF CCF-2007079 and a Simons Investigator Award.
\bibliographystyle{alpha}
\bibliography{MTS}

\end{document}

%% file: preamble/macros.tex
\ifnum\fastmode=0
\usepackage{etex}
\fi

\ifnum\fastmode=0
\usepackage[l2tabu, orthodox]{nag}
\fi

\usepackage{xspace,enumerate}

\usepackage[dvipsnames]{xcolor}

\ifnum\fastmode=0
\usepackage[T1]{fontenc}
\usepackage[full]{textcomp}
\fi

\usepackage[american]{babel}

\usepackage{mathtools}

\usepackage{amsthm}

\newtheorem{theorem}{Theorem}[section]
\newtheorem*{theorem*}{Theorem}

\newtheorem*{proposition*}{Proposition}
\newtheorem{lemma}[theorem]{Lemma}
\newtheorem*{lemma*}{Lemma}
\newtheorem{corollary}[theorem]{Corollary}
\newtheorem*{conjecture*}{Conjecture}

\newtheorem*{fact*}{Fact}

\newtheorem*{exercise*}{Exercise}

\newtheorem*{hypothesis*}{Hypothesis}

\theoremstyle{definition}

\newtheorem{exercise-easy}[theorem]{Exercise}
\newtheorem{exercise-med}[theorem]{Exercise}
\newtheorem{exercise-hard}[theorem]{Exercise$^\star$}
\newtheorem{claim}[theorem]{Claim}
\newtheorem*{claim*}{Claim}

\newtheorem*{remark*}{Remark}
\newtheorem{observation}[theorem]{Observation}
\newtheorem*{observation*}{Observation}

\ifnum\arxivmode=0
\usepackage{newpxtext} 
\usepackage{textcomp} 
\usepackage[varg,bigdelims]{newpxmath}
\usepackage[scr=rsfso]{mathalfa}
\usepackage{bm} 
\let\mathbb\varmathbb
\fi

\ifnum\arxivmode=1
\usepackage[varg]{pxfonts} 
\usepackage{textcomp} 
\usepackage[scr=rsfso]{mathalfa}
\usepackage{bm} 
\fi

\ifnum\showkeys=1
\usepackage[color]{showkeys}
\fi

\makeatletter
\@ifpackageloaded{hyperref}
{}
{\usepackage[pagebackref]{hyperref}}

\definecolor{bleudefrance}{rgb}{0.01, 0.1, 1.0}
\definecolor{azure}{rgb}{0.0, 0.5, 1.0}

\ifnum\showcolorlinks=1
\hypersetup{
pagebackref=true,
colorlinks=true,
urlcolor=blue,
linkcolor=blue,
citecolor=OliveGreen}
\fi

\ifnum\showcolorlinks=0
\hypersetup{
pagebackref=true,
colorlinks=false,
pdfborder={0 0 0}}
\fi

\usepackage{prettyref}

\newcommand{\savehyperref}[2]{\texorpdfstring{\hyperref[#1]{#2}}{#2}}

\newrefformat{eq}{\savehyperref{#1}{\textup{(\ref*{#1})}}}
\newrefformat{lem}{\savehyperref{#1}{Lemma~\ref*{#1}}}
\newrefformat{def}{\savehyperref{#1}{Definition~\ref*{#1}}}
\newrefformat{obs}{\savehyperref{#1}{Observation~\ref*{#1}}}
\newrefformat{thm}{\savehyperref{#1}{Theorem~\ref*{#1}}}
\newrefformat{cor}{\savehyperref{#1}{Corollary~\ref*{#1}}}
\newrefformat{cha}{\savehyperref{#1}{Chapter~\ref*{#1}}}
\newrefformat{sec}{\savehyperref{#1}{Section~\ref*{#1}}}
\newrefformat{app}{\savehyperref{#1}{Appendix~\ref*{#1}}}
\newrefformat{tab}{\savehyperref{#1}{Table~\ref*{#1}}}
\newrefformat{fig}{\savehyperref{#1}{Figure~\ref*{#1}}}
\newrefformat{hyp}{\savehyperref{#1}{Hypothesis~\ref*{#1}}}
\newrefformat{alg}{\savehyperref{#1}{Algorithm~\ref*{#1}}}
\newrefformat{rem}{\savehyperref{#1}{Remark~\ref*{#1}}}
\newrefformat{item}{\savehyperref{#1}{Item~\ref*{#1}}}
\newrefformat{step}{\savehyperref{#1}{step~\ref*{#1}}}
\newrefformat{conj}{\savehyperref{#1}{Conjecture~\ref*{#1}}}
\newrefformat{fact}{\savehyperref{#1}{Fact~\ref*{#1}}}
\newrefformat{prop}{\savehyperref{#1}{Proposition~\ref*{#1}}}
\newrefformat{prob}{\savehyperref{#1}{Problem~\ref*{#1}}}
\newrefformat{claim}{\savehyperref{#1}{Claim~\ref*{#1}}}
\newrefformat{relax}{\savehyperref{#1}{Relaxation~\ref*{#1}}}
\newrefformat{red}{\savehyperref{#1}{Reduction~\ref*{#1}}}
\newrefformat{part}{\savehyperref{#1}{Part~\ref*{#1}}}
\newrefformat{ex}{\savehyperref{#1}{Exercise~\ref*{#1}}}
\newrefformat{property}{\savehyperref{#1}{Property~\ref*{#1}}}

\newcommand{\Sref}[1]{\hyperref[#1]{\S\ref*{#1}}}

\usepackage{nicefrac}

\ifnum\fastmode=0
\ifnum\usemicrotype=1
\usepackage{microtype}
\fi
\fi

\ifnum\showauthornotes=2
\newcommand{\mynotes}[1]{{\sffamily\small\color{teal}{#1}}\medskip}
\newcommand{\Authornote}[2]{{\sffamily\small\color{Maroon}{[#1: #2]}}\medskip}
\newcommand{\Authornotecolored}[3]{{\sffamily\small\color{#1}{[#2: #3]}}}
\newcommand{\Authorcomment}[2]{{\sffamily\small\color{gray}{[#1: #2]}}}
\newcommand{\Authorstartcomment}[1]{\sffamily\small\color{gray}[#1: }

\newcommand{\Authorfnote}[2]{\footnote{\color{red}{#1: #2}}}
\newcommand{\Authorfixme}[1]{\Authornote{#1}{\textbf{??}}}
\newcommand{\Authormarginmark}[1]{\marginpar{\textcolor{red}{\fbox{\Large #1:!}}}}
\newcommand{\myexplain}[1]{{\sffamily\small\color{red}{\noindent [Explanation:\medskip\newline \begin{quote}#1\hfill]\end{quote}}}\medskip}
\newcommand{\explain}[1]{{\sffamily\small\color{red}{#1}}\medskip}

\else

\newcommand{\mynotes}[1]{}
\newcommand{\Authornote}[2]{}
\newcommand{\Authornotecolored}[3]{}
\newcommand{\Authorcomment}[2]{}
\newcommand{\Authorstartcomment}[1]{}

\newcommand{\Authorfnote}[2]{}
\newcommand{\Authorfixme}[1]{}
\newcommand{\Authormarginmark}[1]{}
\newcommand{\myexplain}[1]{}
\newcommand{\explain}[1]{}

\ifnum\showauthornotes=1
\renewcommand{\myexplain}[1]{{\sffamily\small\color{red}{\noindent \begin{quote}{\bf Explanation:} \medskip\newline #1\end{quote}}}\medskip}
\fi

\fi

\ifnum\showfixme=0

\fi

\usepackage{boxedminipage}

\newcommand{\Esymb}{\mathbb{E}}
\newcommand{\Psymb}{\mathbb{P}}

\DeclareMathOperator*{\E}{\Esymb}

\DeclareMathOperator*{\ProbOp}{\Psymb}

\renewcommand{\Pr}{\ProbOp}

\newcommand{\textparen}[1]{\text{(#1)}}

\ifx\because\undefined
\newcommand{\because}[1]{\textparen{because #1}}
\else
\renewcommand{\because}[1]{\textparen{because #1}}
\fi

\newcommand{\seteq}{\mathrel{\mathop:}=}

\newcommand{\bigmid}{~\big|~}

\newcommand\bdot\bullet

\ifx\mathds\undefined 

\else

\fi

\DeclareMathOperator{\poly}{poly}

\DeclareMathOperator{\argmax}{argmax}

\DeclareMathOperator{\dist}{dist}

\newcommand{\R}{\mathbb R}

\newcommand{\cA}{\mathcal A}

\newcommand{\cD}{\mathcal D}

\newcommand{\cF}{\mathcal F}

\newcommand{\cP}{\mathcal P}
\newcommand{\cQ}{\mathcal Q}

\renewcommand{\leq}{\leqslant}

\renewcommand{\geq}{\geqslant}

\ifnum\showdraftbox=1

\else

\fi

\let\epsilon=\varepsilon

\numberwithin{equation}{section}

\newcommand\MYcurrentlabel{xxx}

\newcommand{\MYstore}[2]{%
  \global\expandafter \def \csname MYMEMORY #1 \endcsname{#2}%
}

\newcommand{\MYload}[1]{%
  \csname MYMEMORY #1 \endcsname%
}

\newcommand{\MYnewlabel}[1]{%
  \renewcommand\MYcurrentlabel{#1}%
  \MYoldlabel{#1}%
}

\newcommand{\MYdummylabel}[1]{}

\newcommand{\torestate}[1]{%
  \let\MYoldlabel\label%
  \let\label\MYnewlabel%
  #1%
  \MYstore{\MYcurrentlabel}{#1}%
  \let\label\MYoldlabel%
}

\newcommand{\restatetheorem}[1]{%
  \let\MYoldlabel\label
  \let\label\MYdummylabel
  \begin{theorem*}[Restatement of \prettyref{#1}]
    \MYload{#1}
  \end{theorem*}
  \let\label\MYoldlabel
}

\newcommand{\restatelemma}[1]{%
  \let\MYoldlabel\label
  \let\label\MYdummylabel
  \begin{lemma*}[Restatement of \prettyref{#1}]
    \MYload{#1}
  \end{lemma*}
  \let\label\MYoldlabel
}

\newcommand{\restateprop}[1]{%
  \let\MYoldlabel\label
  \let\label\MYdummylabel
  \begin{proposition*}[Restatement of \prettyref{#1}]
    \MYload{#1}
  \end{proposition*}
  \let\label\MYoldlabel
}

\newcommand{\restatefact}[1]{%
  \let\MYoldlabel\label
  \let\label\MYdummylabel
  \begin{fact*}[Restatement of \prettyref{#1}]
    \MYload{#1}
  \end{fact*}
  \let\label\MYoldlabel
}

\newcommand{\restate}[1]{%
  \let\MYoldlabel\label
  \let\label\MYdummylabel
  \MYload{#1}
  \let\label\MYoldlabel
}

\newcommand{\addreferencesection}{
  \phantomsection
\ifnum\stocmode=0
  \addcontentsline{toc}{section}{References}
\else
  \addcontentsline{toc}{section}{References \hspace*{1in} --------- End of extended abstract ---------}
\fi

}

\newcommand{\e}{\epsilon}
\newcommand{\eps}{\epsilon}

\let\origparagraph\paragraph
\renewcommand{\paragraph}[1]{\vspace*{-3pt}\origparagraph{#1.}}

\allowdisplaybreaks

\usepackage{paralist}

\usepackage{comment}

%% file: preamble/macros-notes.tex
\let\pref=\prettyref

\newcommand{\diam}{\mathrm{diam}}
\newcommand{\dmid}{\,\|\,}

\newcommand{\vertiii}[1]{{\left\vert\kern-0.25ex\left\vert\kern-0.25ex\left\vert #1 
          \right\vert\kern-0.25ex\right\vert\kern-0.25ex\right\vert}}

\newcommand\f{\varphi}

\DeclareMathOperator{\argmin}{\mathrm{argmin}}

